\documentclass[american,aps,pra,reprint, superscriptaddress, showpacs]{revtex4-2}
\usepackage[unicode=true,pdfusetitle, bookmarks=true,bookmarksnumbered=false,bookmarksopen=false, breaklinks=false,pdfborder={0 0 0},backref=false,colorlinks=false] {hyperref}
\hypersetup{ colorlinks,linkcolor=myurlcolor,citecolor=myurlcolor,urlcolor=myurlcolor}
\usepackage{graphics,epstopdf,graphicx, amsthm, amsmath, amssymb, times, braket, color, bm, physics,comment,graphicx}

\definecolor{myurlcolor}{rgb}{0,0,0.7}
\newcommand*{\Scale}[2][4]{\scalebox{#1}{$#2$}}%

\newtheorem{theorem}{Theorem}
\newtheorem{definition}{Definition}
\newtheorem{proposition}{Proposition}
\newtheorem{lemma}{Lemma}

\newcommand{\iden}{\mathbb{I}}

\newcommand{\spanned}{\mathrm{span}}
\newcommand{\mH}{\mathcal{H}}
\newcommand{\mL}{\mathcal{L}}
\newcommand{\mT}{\mathcal{T}}
\newcommand{\mS}{\mathcal{S}}
\newcommand{\mB}{\mathcal{B}}
\newcommand{\mR}{\mathcal{R}}
\newcommand{\mC}{\mathcal{C}}
\newcommand{\mV}{\mathcal{V}}

\begin{document}


\title{Masking quantum information on hyperdisks}


\author{Feng Ding}
\author{Xueyuan Hu}
\email{xyhu@sdu.edu.cn}

\affiliation{School of Information Science and Engineering, Shandong University, Qingdao 266237, China}


\date{\today}

\begin{abstract}
    Masking information is a protocol that encodes quantum information into a bipartite entangled state while the information is completely unknown to local systems.
    This paper explicitly studies the structure of the set of maskable states and its relation to hyperdisks.
    We prove that although the qubit states which can be masked must locate on a single hyperdisk, the set of maskable states can consist of two or more hyperdisks for high-dimensional cases.
    Our results may shed light on several research fields of quantum information theory, such as the structure of entangled states and the local discrimination of bipartite states.
\end{abstract}


\maketitle

\section{Introduction}
There are a variety of no-go theorems that characterize the intrinsic gap between classical and quantum information,
such as the no-cloning theorem \cite{cloning}, the no-deleting theorem \cite{Kumar2000}, and the no-go theorem for creating the superposition of unknown states \cite{superposition}.
A branch of no-go theorems is related to entanglement such as the no-hiding theorem \cite{hiding}.

Recently, the authors of Ref. \cite{masking} proposed a masking quantum information protocol,
which encodes quantum information into a bipartite entangled system, while the information is completely unknown to local systems.
They derived a new no-go theorem called the no-masking theorem, which claims that although one can encode classical information into entanglement, masking arbitrary quantum states is impossible.
Still, one can go beyond the classical world and mask a set of nonorthogonal quantum states into bipartite states.
Furthermore, Ref. \cite{tri} generalized the protocol and proved that it is possible to mask full quantum information into multipartite systems, Ref. \cite{PhysRevA.99.052343} developed a probabilistic masking protocol, and Ref. \cite{PhysRevA.100.030304} gave a characterization of maskable qubit states.

The structure of maskable states helps us to gain better understanding of the classification of high-dimensional entangled states \cite{RevModPhys.81.865,belltest15}.
Because the bipartite entangled target states are fully indistinguishable by two participants who are forbidden to communicate, the task of masking information is related to the research on the local discrimination task \cite{PhysRevA.91.052314,PhysRevA.98.022304}.
Notice that masking information can be viewed as a quantum secret sharing scheme \cite{PhysRevA.59.162,PhysRevA.59.1829,PhysRevA.61.042311},
so it is significant to study the structure of maskable states as the shareable quantum secrets.
Since it is impossible to mask all the quantum states, the authors of Ref. \cite{masking} designed a masker using the generalized controlled-NOT gate.
Based on this masker, they proposed a hyperdisk conjecture, which said that any set of maskable states must live in some disk.

In this paper, we prove that the hyperdisk conjecture holds for the qubit case, while it fails for the higher-dimensional case.
For this purpose, we first give a clear definition of the hyperdisk and introduce some related concepts.
Then we study the classification of the masking protocol, depending on the dimension $n$ of the input space, the Schmidt number $d$ of target states, and the degeneracy of marginal states.
General methods are provided to derive the structure of maskable states in different cases.
Based on these methods, we show that the maskable states may live in two or more different hyperdisks if $n\geq 3$.
Full characterizations of the sets of maskable states for $n=2,d\geq 2$ and for $n=3,d=3$ are given in the last section.

\section{Hyperdisk and related concepts}
Let $\mH$ be an $n$-dimensional Hilbert space and let $\mB=\left\{\ket{\phi_j}\right\}^{m-1}_{j=0}$ be an orthonormal basis of an $m$-dimensional subspace of $\mH$.
We introduce a real vector $\bm{r}_\mB$ for each pure state $\ket{\psi}\in\mH$ as follows:
\begin{equation}\label{eq:rv}
    \bm{r}_\mB(\ket{\psi})=\left(\left|\bra{\phi_0}\ket{\psi}\right|,...,\left|\bra{\phi_{m-1}}\ket{\psi}\right|\right)^\mathrm{T}.
\end{equation}
Notice that $\bm{r}_\mB(\ket{\Psi})$ is normalized if and only if $\ket{\psi}\in\spanned\{\mB\}$.

\begin{definition}[hyperdisk]
    Let $\mS$ be a set of pure states in an $n$-dimensional Hilbert space $\mH$.
    Then $\mS$ is a hyperdisk if there is a complete orthonormal basis $\mathcal{B}$ of $\mV_\mS:=\spanned\{\mathcal{S}\}$ such that
    \begin{eqnarray}
        \bm{r}_\mB(\ket{\psi})=\bm{r},&&\quad\forall\ket{\psi}\in\mS,\\
        \bm{r}_\mathcal{B}\left(\ket{\xi}\right)\neq\bm{r},&&\quad\forall\ket{\xi}\in \mH\setminus\mS.
    \end{eqnarray}
    where $\bm{r}$ is a constant vector with strictly positive entries.
\end{definition}

Here we call $\mB$ the hyperdisk basis, $m=\mathrm{dim}(\mV_\mS)$ the dimension of the hyperdisk, and $\bm{r}$ the coefficient vector.
In the following, we will use the mathcal typeface $\mathcal{X}$ to denote a set of pure states and $\mV_\mathcal{X}$ to denote the subspace spanned by $\mathcal{X}$, i.e., $\mV_\mathcal{X}:=\spanned\{\mathcal{X}\}$. Also, the dimension of $\mathcal{X}$ is labeled as $\mathrm{dim}\left(\mV_\mathcal{X}\right)$.

In the trivial case with $m=1$, a hyperdisk consists of only one pure state.
For $m=2$, a hyperdisk can be expressed as
\begin{equation}
    \left\{\ket{\psi(\theta)}=a\ket{\phi_0}+be^{i\theta}\ket{\phi_1}|\theta\in\mathbb{R}\right\},
\end{equation}
where $\left\{\ket{\phi_0},\ket{\phi_1}\right\}$ is the hyperdisk basis and $a,b$ are positive real numbers.
In the Bloch representation, a two-dimensional hyperdisk can be visualized as an intersection of the sphere and a complex plane.
The plane is orthogonal to the crossing line of antipodal points $\ket{\phi_0}$ and $\ket{\phi_1}$ (Fig.\ref{fg:bl}).
Furthermore, in general, any pure state in an $m$-dimensional hyperdisk $\mathcal{S}$ can be written as
\begin{equation}\label{eq:dh}
    \ket{\psi\left(\bm{\theta}\right)}=\sum^{m-1}_{j=0}r_j e^{i\theta_j}\ket{\phi_j},\quad \bm{\theta}\in\mathbb{R}^m,
\end{equation}
where $\mB=\{\ket{\phi_j}\}$ is the hyperdisk basis of $\mS$. In the following, we will use the expression as in Eq. (\ref{eq:dh}) to represent a set of pure states, i.e., $\ket{\psi\left(\bm{\theta}\right)}$ denotes the set of states $\left\{\ket{\psi(\bm{\theta})}|\bm{\theta}\in \mathbb{R}^m\right\}$.
\begin{figure}[htbp]
    \centering
    \includegraphics[height=4.5cm, width=4.75cm]{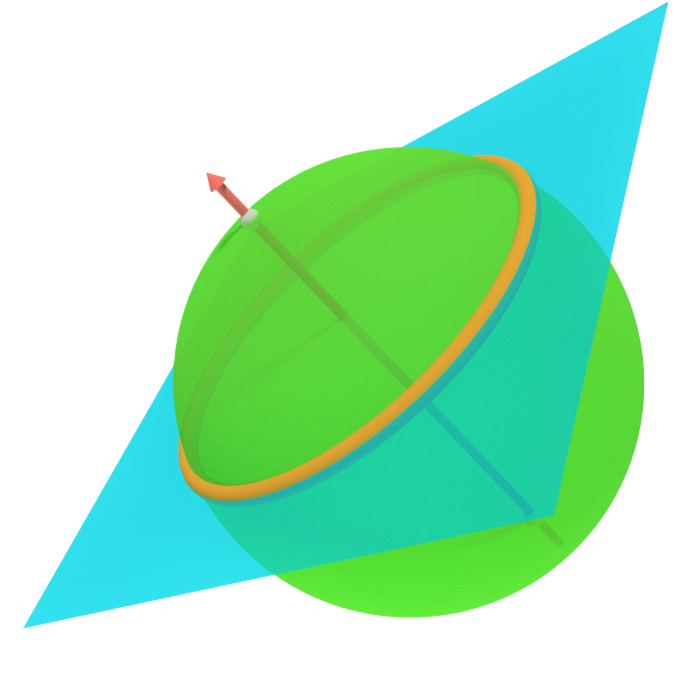}
    \caption{An illustration of the geometric interpretation of a two-dimensional hyperdisk. A two-dimensional hyperdisk is a circle on the Bloch sphere, which is represented as the intersection between the Bloch sphere and a plane. The cross line between the states in hyperdisk basis is perpendicular to the plane. }
    \label{fg:bl}
\end{figure}

Every hyperdisk satisfies the following properties.

\noindent [\textbf{Property 1}]\label{pty:1}.
One can generate all states in hyperdisk $\mS=\left\{\ket{\psi(\bm{\theta})}\right\}_{\bm{\theta}}$ by applying a set of commutative unitary operators $\left\{U(\bm{\theta})\right\}_{\bm{\theta}}$ on an arbitrary fixed state $\ket{\psi_0}\in\mS$, i.e., 
\begin{equation}\label{cmd}
        \ket{\psi(\bm{\theta})}=U(\bm{\theta})\ket{\psi_0},
\end{equation}
where $[U(\bm{\theta}),U(\bm{\theta}')]=0,\forall\bm{\theta}\neq\bm{\theta}'$.

In order to prove this property, we construct the set of commutative unitaries in Eq. (\ref{cmd}) as
\begin{equation}
    U(\bm{\theta})=\sum^{m-1}_{j=0}e^{i\theta_j}\ket{\phi_j}\bra{\phi_j}+\Pi^\perp,
\end{equation}
where $\left\{\ket{\phi_j}\right\}_j$ is the hyperdisk basis of $\mS$ and $\Pi^\perp$ is the projection operator to the orthogonal complement space of $\mV_\mS$.

\noindent [\textbf{Property 2}]\label{pty:2}. A linear isometry $V:\mH\rightarrow\mH^\prime$ preserves the hyperdisk structure. 

Notice that the condition $\mathrm{dim}(\mH^\prime)\geq \mathrm{dim}(\mH)$ is implied from the definition of linear isometry.
In the following, we prove this property. 
Let $\mS$ be a hyperdisk in $\mH$ and let $\mB=\left\{\ket{\phi_j}\right\}_j$ be the hyperdisk basis of $\mS$.
After the action of $V$, each state $\ket{\psi}\in\mS$ becomes $\ket{\psi'}=V\ket{\psi}\in\mS'$, and the hyperdisk basis $\mB$ becomes $\mB^\prime=\left\{V\ket{\phi_j}\right\}_j$. 
Clearly, $\mB^\prime$ is a set of orthonormal states.
For any state  $\ket{\psi'}\in\mS'$, the real vector 
\begin{equation*}
    \begin{aligned}
    \bm{r}_{\mathcal{B}^\prime}(\ket{\psi'})&=\left(\left|\bra{\phi_0}V^\dagger V\ket{\psi}\right|,...,\left|\bra{\phi_{m-1}}V^\dagger V\ket{\psi}\right|\right)^\mathrm{T}\\
    &=\left(\left|\bra{\phi_0}\ket{\psi}\right|,...,\left|\bra{\phi_{m-1}}\ket{\psi}\right|\right)^\mathrm{T}\\
    &=\bm{r}_{\mathcal{B}}(\ket{\psi})
    \end{aligned}
\end{equation*}
is constant.
It follows that $\mS'\subset\mH'$ is a hyperdisk with the same dimension and the coefficient vector as $\mS$.

\noindent [\textbf{Property 3}]\label{pty:3}. Every pair of pure states lives in some hyperdisk.

Notice that any pair of states spans a two-dimensional qubit space.
Geometric interpretation (Fig.\ref{fg:bl}) shows that every two points on the Bloch sphere can live in the same hyperdisk.

For further discussion on the masking protocol in Section III, the hyperdisks in the space of the bipartite system should be taken into consideration, which leads us to the concept of the Schmidt hyperdisk.
\begin{definition}[Schmidt hyperdisk]\label{def:sh}
    A Schmidt hyperdisk $\mS^{AB}$ in Hilbert space $\mH_{AB}=\mH_A^{(d)}\otimes\mH_B^{(d)}$ is expressed as
    \begin{equation}\label{eq:schmhyper}
        \ket{\Psi(\bm{\theta})}=\sum_{j=0}^{d-1}r_j e^{i\theta_j}\ket{\phi^A_j\phi^B_j},\ \bm{\theta}\in\mathbb{R}^d,
    \end{equation}
    where $r_j\neq0$, and $\left\{\ket{\phi^{A,B}_j}\right\}_{j=0}^{d-1}$ is an orthonormal basis of $\mH_{A,B}^{(d)}$.
\end{definition}
The state in Eq. (\ref{eq:schmhyper}) is in the Schmidt decomposition form \cite{10uni}. The set $\left\{\ket{\phi^A_j\phi^B_j}\right\}_{j=0}^{d-1}$ is usually called a Schmidt basis. Here we stress that a hyperdisk is a Schmidt hyperdisk only if its basis is a Schmidt basis. For example, the hyperdisk $\ket{\Psi^\prime(\theta)}=\frac{1}{\sqrt{3}}[\ket{00}+e^{i\theta}(\ket{11}+\ket{22})]$ is not a Schmidt hyperdisk, because its basis is $\left\{\ket{00},\frac{1}{\sqrt2}(\ket{11}+\ket{22})\right\}$ and $\frac{1}{\sqrt2}(\ket{11}+\ket{22})$ does not belong to a Schmidt basis. Nevertheless, $\ket{\Psi^\prime(\theta)}$ is a subset of the Schmidt hyperdisk $\ket{\Psi(\theta_1,\theta_2)}=\frac{1}{\sqrt{3}}(\ket{00}+e^{i\theta_1}\ket{11}+e^{i\theta_2}\ket{22})$.
This leads us to the concept of subhyperdisk. This concept helps us to explore the sub-structures of a hyperdisk.

\begin{definition}[subhyperdisk]
    Let $\mS$ be a hyperdisk. A subset $\mS^\prime\subseteq\mS$ is a subhyperdisk of $\mS$ , if $\mS^\prime$ is also a hyperdisk.
\end{definition}

Here, we derive the general form of a subhyperdisk. Consider an $m$-dimensional hyperdisk $\mS$ with basis $\mB=\{\ket{\phi_j}\}_j$ and coefficient vector $\bm{r}$ (with $j$th entry $r_j$), and an $m^\prime$-dimensional hyperdisk $\mS^\prime$ ($m'\leq m$) with basis $\mB'=\{\ket{\phi'_k}\}_k$ and coefficient vector $\bm{r'}$ (with $k$th entry $r'_k$).
We define a Gramian matrix as
\begin{equation}
    G_{jk}=r'_k\bra{\phi_j}\ket{\phi^\prime_k}.
\end{equation}
If $\mS'$ is a subhyperdisk of $\mS$, then each row of $G$ has exactly one nonzero entry with absolute value $r_j$ and each column of $G$ contains at least one nonzero entry. The reason is as follows.
Any state $\ket{\psi^\prime(\bm{\theta})}\in\mS^\prime$ is expressed as $\ket{\psi^\prime(\bm{\theta})}=\sum_{k=0}^{m'-1}r'_k e^{i\theta_k}\ket{\phi'_k}$.
Because $\ket{\psi^\prime(\bm{\theta})}\in\mS^\prime\subseteq\mS$,
we have 
\begin{equation}
    r_j=|\bra{\phi_j}\ket{\psi^\prime(\bm{\theta})}|=\left|\sum^{m^{\prime}-1}_{k=0} e^{i\theta_k} G_{jk}\right|,\quad\forall \bm{\theta}\in\mathbb{R}^{m^\prime}.
\end{equation}
The summation has constant absolute value for all $\bm{\theta}$ only if there is exactly one nonzero term in $\left\{G_{jk}\right\}_k$. Further, because $\ket{\phi'_k}\in\mV_{\mS'}\subseteq\mV_\mS=\spanned\left\{\ket{\phi_j}\right\}_j$ and $r'_k\neq0$, there is at least one nonzero entry in the column of $G$.

We now turn to the relation between subhyperdisks. To determine whether two subhyperdisks are contained in a single hyperdisk, we give the proposition below as a criterion.
\begin{proposition}\label{pp:n}
    Suppose there are two $n$-dimensional hyperdisks $\mS_0$ and $\mS_1$ with hyperdisk bases $\mB_0$ and $\mB_1$, respectively, where $m\geq 2$. If $\mS_0$ and $\mS_1$ are subsets of a single $(m+1)$-dimensional hyperdisk, then there exist two states $\ket{\phi^0_k}\in\mB_0$ and $\ket{\phi^1_l}\in\mB_1$ such that 
    \begin{equation}\label{eq:qqq}
        \bra{\phi^0_k}\ket{\phi^1_l}= 0.
    \end{equation}
\end{proposition}
\begin{proof}
    The Gramian matrix $G^0$ of subhyperdisk $\mS^0$ is a $(m+1)\times m$ matrix, which has exactly one nonzero entry in each row and at least one nonzero entry in each column. By properly ordering the hyperdisk bases $\mB$ and $\mB_0$, the entries $G^0_{jk}=r^0_k\bra{\phi_j}\ket{\phi^0_k}$ are $G^0_{00}=r_0$, $G^0_{10}=r_1e^{i\varphi}$, and $G^0_{j0}=0$ for $j\neq0,1$, and, for $1\leq k\leq m-1$, $G^0_{jk}=r_j\delta_{j,k+1}$.
    It follows that 
    \begin{eqnarray}
        r_0^0\ket{\phi^0_0}&=&r_0\ket{\phi_0}+r_1e^{i\varphi}\ket{\phi_1},\nonumber\\
        \ket{\phi^0_k}&=&\ket{\phi_{k+1}},\ 1\leq k\leq m-1.
    \end{eqnarray}
    
    Similarly, there exists a state in $\mB_1$, labeled as $\ket{\phi^1_0}$, which is a linear combination of two states in $\mB$:
    \begin{equation}
        \ket{\phi^1_0}=c_a\ket{\phi_a}+c_b\ket{\phi_b},
    \end{equation}
    where $c_a,c_b\neq 0$ are complex coefficients, and without loss of generality we set $a<b$. Further, we have $\mB_1\backslash\left\{\ket{\phi_0^1}\right\}=\mB\backslash\left\{\ket{\phi_a},\ket{\phi_b}\right\}$. Notice that these two sets are not null due to $m\geq2$.
    
    \noindent \textbf{[Case 1]} $\left|\bra{\phi^0_0}\ket{\phi^1_0}\right|=0$. Then Eq. (\ref{eq:qqq}) holds for $k=l=0$.

    \noindent \textbf{[Case 2]} $\left|\bra{\phi^0_0}\ket{\phi^1_0}\right|=1$. Then Eq. (\ref{eq:qqq}) holds for $k=0,l\neq0$.

    \noindent \textbf{[Case 3]} $\left|\bra{\phi^0_0}\ket{\phi^1_0}\right|\neq 0,1$. Then $a=0$ or $1$. 
    Suppose $a=0,b=1$; then Eq. (\ref{eq:qqq}) holds for $k=0,l\neq 0$. 
    Suppose $a=0,b>1$; we then have $\ket{\phi_1^0}=\ket{\phi_2}\in\mB_0$  and $\ket{\phi_{l_1}^1}=\ket{\phi_1}\in\mB_1$ such that Eq. (\ref{eq:qqq}) holds for $k=1,l=l_1$. 
    Suppose $a=1,b>1$; we then have $\ket{\phi_1^0}=\ket{\phi_2}\in\mB_0$  and $\ket{\phi_{l_0}^1}=\ket{\phi_0}\in\mB_1$ such that Eq. (\ref{eq:qqq}) holds for $k=1,l=l_0$. 
    
    To sum up, there exist indices $k$ and $l$ for all the three cases such that Eq. (\ref{eq:qqq}) holds.
\end{proof}

In the following, we define the regular subset of a hyperdisk. This concept is essential to the characterization of maskable states of a nondegenerate masking machine, which will be discussed in Sec. \ref{sec:ndc}.
\begin{definition}[regular subset of hyperdisk]\label{ref:rs}
    Let $\mS$ be a hyperdisk. 
    The set $\mC$ is a regular subset of $\mS$ if
    \begin{equation}\label{eq:rd}
        \mV_\mC\cap\mS=\mC,
    \end{equation}
    where $\mV_\mC=\spanned\{\mC\}$.
\end{definition}

Eq. (\ref{eq:rd}) can be interpreted as follows. 
If $\ket{\eta}$ is a complex linear combination of states in $\mC$, i.e., $\ket{\eta}\in\mV_\mC$,
then the condition $\ket{\eta}\in\mS$ is equivalent to $\ket{\eta}\in\mC$.
Notice that every subhyperdisk is also a regular subset.

A key property of the regular subset is 
\begin{equation}\label{eq:reg_eqiv}
    \dim(\mC)=\dim(\mS)\Leftrightarrow\mC=\mS.
\end{equation}
The reason is as follows. The condition $\mC\subseteq\mS$ implies that $\mV_\mC$ is a subspace of $\mV_\mS$. Then $\dim(\mC)=\dim(\mS)$ is equivalent to $\mV_\mC=\mV_\mS$. It is in turn equivalent to $\mC=\mS$, because $\mC=\mV_\mC\cap\mS=\mV_\mS\cap\mS=\mS$.

A general subset $\mathcal{G}$ of hyperdisk $\mS$ can be expressed as
\begin{equation}\label{gc}
    \mathcal{G}=\bigcup_{p\in\mathcal{P}}\mS_p,
\end{equation}
where $\left\{\mS_p|p\in\mathcal{P}\right\}$ is the set of all hyperdisks contained in $\mathcal{G}$.
Notice that this expression is valid because every single pure state forms a one-dimensional hyperdisk.
Nevertheless, this formulation does not limit the number of hyperdisks in $\mathcal{G}$.
If $\mathcal{G}$ consists of a finite number of hyperdisks, we define the optimal cover number of $\mathcal{G}$ as the least number of hyperdisks that we need to fully cover $\mathcal{G}$.
The following lemma implies that the optimal cover number of a two-dimensional regular subset of a general finite dimensional hyperdisk is at most 2.
\begin{lemma}\label{lm:c2}
    A two-dimensional regular subset of a general finite-dimensional hyperdisk is either a set of two distinct pure states or a two-dimensional hyperdisk.
\end{lemma}
\begin{proof}
    Let $\mS$ be an $m$-dimensional hyperdisk in the form of Eq. (\ref{eq:dh}) and let $\mC$ be its two-dimensional regular subset. By definition, there are at least two states in $\mC$. Without loss of generality, these two states can be written as  
    \begin{equation}\label{eq:lc}
        \ket{\psi_0}=\sum^{m-1}_{j=0}r_j\ket{\phi_j},\quad \ket{\psi_1}=\sum^{m-1}_{j=0}r_je^{i\theta_j}\ket{\phi_j},
    \end{equation}
    where $\mB=\left\{\ket{\phi_j}\right\}_j$ is the hyperdisk basis of $\mS$, and the phases $\theta_j$ are not equal (because otherwise we would have $\ket{\psi_0}=\ket{\psi_1}$ up to a phase factor). Notice that $\mV_\mC=\spanned\left\{\ket{\psi_0},\ket{\psi_1}\right\}$.
    Then any state in $\mC$ can be expressed as 
    \begin{equation}\label{eq:lcab}
        \ket{\psi}=a\ket{\psi_0}+be^{i\varphi}\ket{\psi_1},
    \end{equation}
    where $a\geq 0,b\geq 0$ and $\varphi$ are chosen such that $\ket{\psi}$ is normalized.
    Recalling $\ket{\psi}\in\mC\subseteq\mS$, we get 
    \begin{equation}\label{eq:21}
        \left|\bra{\phi_j}\ket{\psi}\right|= r_j,\quad\forall j.
    \end{equation}
    By substituting Eqs. (\ref{eq:lc}) and (\ref{eq:lcab}) into Eq. (\ref{eq:21}), we arrive at
    \begin{equation}\label{eq:22}
        a^2+b^2+2ab\cos(\varphi+\theta_j)= 1,\quad \forall j.
    \end{equation}
    
        \noindent\textbf{[Case 1]} There exists $j\neq j^\prime$ such that $\cos(\varphi+\theta_j)\neq\cos(\varphi+\theta_{j^\prime})$. Then the only two solutions to Eq. (\ref{eq:22}) are $a=0,b=1$ and $a=1,b=0$. This means that any state in $\mC$ is either $\ket{\psi_0}$ or $\ket{\psi_1}$, i.e., $\mC=\{\ket{\psi_0},\ket{\psi_1}\}$. In this case, $\mC$ consists of exactly two pure states.

        \noindent\textbf{[Case \ 2]} $\cos(\varphi+\theta_j)=\cos(\eta)$ for all $j$, where $\eta$ is a constant parameter. Without loss of generality, we set $\varphi+\theta_j\in[-\pi,\pi)$ and $\eta\in[0,\pi]$, and then
        \begin{equation}
            \varphi+\theta_j=\pm\eta.
        \end{equation}
    Further, because it is required that $\theta_j$ are not equal, we have $\eta\neq0,\pi$.
    Thus $\ket{\psi_1}$ can be reformulated as 
    \begin{align}
        \begin{aligned}
        e^{i\varphi}\ket{\psi_1}&=e^{-i\eta}\sum_{j:\theta_j+\varphi=-\eta}r_j\ket{\phi_j}+e^{i\eta}\sum_{j:\theta_j+\varphi=\eta}r_j\ket{\phi_j}\\
        &=e^{-i\eta}r'_-\ket{\phi^\prime_-}+e^{i\eta}r'_+\ket{\phi^\prime_+},
        \end{aligned}\raisetag{20pt}
    \end{align}
    where $r'_\pm\ket{\phi^\prime_\pm}\equiv\sum_{j:\theta_j+\varphi=\pm\eta}r_j\ket{\phi_j}\neq0$.
    Similarly, we have $\ket{\psi_0}=r'_-\ket{\phi_-^\prime}+r'_+\ket{\phi_+^\prime}$.  
    Therefore, $\mV_\mC=\spanned\{\ket{\phi^\prime_-},\ket{\phi^\prime_+}\}$.
    Using the condition $\mC=\mV_C\cap\mS$, we get
    \begin{equation}
        \mC=\left\{\ket{\psi(\theta)}=r'_-\ket{\phi^\prime_-}+e^{i\theta}r'_+\ket{\phi^\prime_+}|\theta\in\mathbb{R}\right\}.
    \end{equation}
    This means that $\mC$ is a two-dimensional subhyperdisk of $\mS$. 
\end{proof}

\section{Masking Information Protocol}
A masking information protocol involves three participants: a referee $R$ and two players $A$ and $B$.
Each of them holds a system with Hilbert spaces $\mH_R$, $\mH_A$, and $\mH_B$, respectively.
In every round of the protocol, the referee randomly chooses a pure state $\ket{\psi}$ in the set of maskable states $\mR\subset\mH_R$,
and loads $\ket{\psi}$ into a masking machine.
\begin{definition}[masking machine]\label{def:masking}
    Let $\mR$ be a set of states in $\mH_R$.
    A masking machine for $\mR$ is a linear isometry $V_\mathrm{mask}:\mH_R\rightarrow\mH_A\otimes\mH_B$ which satisfies the following two conditions.
    
    $\textnormal{(1) }$  $\forall\ket{\psi}\in\mathcal{R}$; then the marginal states of $\ket{\Psi}=V_{\mathrm{mask}}\ket{\psi}$ read
    \begin{equation}
        \Tr_B(\ket{\Psi}\bra{\Psi})=\rho_A\ \textnormal{and}\ \Tr_A(\ket{\Psi}\bra{\Psi})=\rho_B,
    \end{equation}
    where $\rho_A$ and $\rho_B$ are independent of $\ket{\psi}$.
    
    $\textnormal{(2) }$ $\forall\ket{\psi'}\notin\mathcal{R}$; then the marginal states of $\ket{\Psi'}=V_{\mathrm{mask}}\ket{\psi'}$ satisfy
    \begin{equation}
        \Tr_B(\ket{\Psi'}\bra{\Psi'})\neq\rho_A\ \textnormal{or} \ \Tr_A(\ket{\Psi'}\bra{\Psi'})\neq\rho_B.
    \end{equation}
    Here, the set $\mR$ is called the set of maskable states. The set $\mT\equiv V_{\mathrm{mask}}\mR\subset\mH_A\otimes\mH_B$ is called the set of target states.
    The dimension of input space is denoted by $n\equiv\dim\left(\mH_R\right)$. 
    The rank of marginal states is denoted by $d\equiv\rank\left(\rho_{A/B}\right)$.
\end{definition}

Here we mention that the masking machine in Ref. \cite{masking} is defined as a bipartite unitary transformation $\ket{\Psi}=U_\mathrm{mask}\ket{\psi}\ket{s}$,
where $\ket{s}$ is a fixed state of  the auxiliary system $\mH_S$.
This is a special case of our definition with $V_\mathrm{mask}=U_\mathrm{mask}\iden_R\otimes\ket{s}$. The advantage of our definition is that we require fewer parameters to fully describe a masking machine.


Our main task in this paper is to study the structure of $\mR$, as well as its relation to and difference from hyperdisks. Because $\mR$ is isomorphic to $\mT$, the rest of this paper will be focused on the structure of $\mT$. Here we first give some notations.

Without loss of generality, we set $\mH_R=\spanned\{\mR\}$, because $\mR$ contains all of the pure states that can be masked and states not in $\spanned\{\mR\}$ are irrelevant.  From the isomorphic relation between $\mR$ and $\mT$, $\mV_\mT\equiv\spanned\{\mT\}$ is isomorphic to $\mH_R$, and hence $\dim(\mT)= n$.

When the marginal states $\rho_A$ and $\rho_B$ are fixed, the set of legal states $\mL$ is defined as the set of all bipartite pure states with marginal states $\rho_A$ and $\rho_B$, i.e., $\mL=\{\ket{\Psi}|\Tr_B(\ket{\Psi}\bra{\Psi})=\rho_A,\Tr_A(\ket{\Psi}\bra{\Psi})=\rho_B\}$. Hence, the legal states must have Schmidt number $\mathrm{Sch}\left(\ket{\Psi}\right)$ equal to $d=\rank\left(\rho_{A/B}\right)$. For simplicity, we also set the local dimension $\dim\left(\mH_{A/B}\right)$ as $d$. Thus $\dim(\mH_{AB})=d^2$.

By definition, a state is in the set of target state $\mT$ if and only if it satisfies the following two conditions:\\
(1) It belongs to $\mL$;\\
(2) It can be mapped to a state in $\mH_R$ by a linear isometry, i.e., it belongs to $\mV_\mT$.\\
Thus, the set of target states can be expressed as
\begin{equation}\label{eq:tl}
    \mT = \mV_\mT\cap\mL.
\end{equation}
This expression is essential to our discussion on the masking protocol. The degeneracy of the marginal states determines the structure of $\mL$.
In the following, we first study the nondegenerate case and the completely degenerate case, and then derive some results for the general case.


\subsection{Nondegenerate Case}\label{sec:ndc}

In this case, the marginal states can be written as
\begin{equation}\label{eq:ts}
    \rho_A=\sum^{d-1}_{j=0}\lambda_j\ket{\phi^A_j}\bra{\phi^A_j}, \rho_B=\sum^{d-1}_{j=0}\lambda_j\ket{\phi^B_j}\bra{\phi^B_j},
\end{equation}
where $\lambda_i\neq\lambda_j,\ \forall i\neq j$. By the purification process, the set of legal states $\mL_{\mathrm{ND}}$ is a $d$-dimensional Schmidt hyperdisk in the following form:
\begin{equation}\label{eq:schm}
    \ket{\Psi\left(\bm{\theta}\right)}=\sum^{d-1}_{j=0}\sqrt{\lambda_{j}}e^{i\theta_j}\ket{\phi^A_j \phi^B_j}.
\end{equation}
Hence $\dim\left(\mV_{\mL\mathrm{ND}}\right)=d$.
By Eq. (\ref{eq:tl}), the set of target states $\mT_\mathrm{ND}$ is a regular subset of this hyperdisk, i.e., $\mT_{\mathrm{ND}}=\mV_{T\mathrm{ND}}\cap\mL_{\mathrm{ND}}$.
Thus the dimension $n$ of input space is bounded as
\begin{equation}\label{ineq}
    n=\mathrm{dim}\left(\mathcal{V}_{\mT\mathrm{ND}}\right)\leq \mathrm{dim}\left(\mV_{\mL\mathrm{ND}}\right)=d.
\end{equation}
From Eq. (\ref{eq:reg_eqiv}), the equality holds if and only if $\mT_\mathrm{ND}=\mL_{\mathrm{ND}}$, which implies that the set of maskable states $\mR_\mathrm{ND}$ is a $d$-dimensional hyperdisk with $n=d$.

When $n<d$, the regular subset $\mT_\mathrm{ND}$ can consist of multiple hyperdisks. Furthermore, the set of maskable states $\mR_\mathrm{ND}$ may not live in a single hyperdisk. 
For example, we consider the following nondegenerate masking protocol with $n=3,d=4$.
Here, $\mT_\mathrm{ND}$ consists of the following two different subhyperdisks of the same Schmidt hyperdisk:
\begin{equation}
    \begin{aligned}
        \ket{\Psi_0(\alpha)}&=\ket{00}+\sqrt{2}\ket{11}+e^{i\alpha}\left(\sqrt{3}\ket{22}+2\ket{33}\right),\\
        \ket{\Psi_1(\beta)}&=\ket{00}+\sqrt{3}\ket{22}+e^{i\beta}\left(\sqrt{2}\ket{11}+2\ket{33}\right).\\
    \end{aligned}
\end{equation}
It follows that $\mV_{\mT\mathrm{ND}}$ is a three-dimensional subspace of $\mH_{AB}$:
\begin{equation}
    \mV_{\mT\mathrm{ND}}=\spanned\{\ket{00}+\sqrt{2}\ket{11},\sqrt{3}\ket{22}+2\ket{33},\ket{\Phi_\perp}\}
\end{equation}
where $\ket{\Phi_\perp}=\frac{2}{3}\ket{00}+\frac{4\sqrt{3}}{7}\ket{22}-\frac{\sqrt{2}}{3}\ket{11}-\frac{6}{7}\ket{33}$.
Here we define the masking machine $V_\mathrm{mask}$ as
$\ket{0}\rightarrow\frac{1}{\sqrt{3}}(\ket{00}+\sqrt{2}\ket{11})$, $\ket{1}\rightarrow\frac{1}{\sqrt{7}}(\sqrt{3}\ket{22}+2\ket{33})$ and $\ket{2}\rightarrow\sqrt{\frac{50}{21}}\ket{\Phi_\perp}$.
Then the corresponding set of maskable states $\mR_\mathrm{ND}$ consists of two two-dimensional hyperdisks $\mS_0$ and $\mS_1$ in the following form:
\begin{equation}\label{eq:ndm}
    \begin{aligned}
        \ket{\psi_0(\alpha)}&=\sqrt{3}\ket{0}+e^{i\alpha}\sqrt{7}\ket{1},\\
        \ket{\psi_1(\beta)}&=\frac{1}{\sqrt{3}}\ket{0}+\frac{3}{\sqrt{7}}\ket{1}-\sqrt{\frac{50}{21}}\ket{2}+\\
        &e^{i\beta}\left(\frac{2}{\sqrt{3}}\ket{0}+\frac{4}{\sqrt{7}}\ket{1}+\sqrt{\frac{50}{21}}\ket{2}\right),\\
    \end{aligned}
\end{equation}
The hyperdisk bases of these two hyperdisks are $\mB_0=\left\{\ket{0},\ket{1}\right\}$ and $\mB_1=\big\{\frac{1}{2}(\frac{1}{\sqrt{3}}\ket{0}+\frac{3}{\sqrt{7}}\ket{1}-\sqrt{\frac{50}{21}}\ket{2}),\frac{1}{\sqrt{6}}(\frac{2}{\sqrt{3}}\ket{0}+\frac{4}{\sqrt{7}}\ket{1}+\sqrt{\frac{50}{21}}\ket{2})\big\}$, respectively. From Proposition \ref{pp:n}, $\mS_0$ and $\mS_1$ are not subsets of a single three-dimensional hyperdisk. Therefore, the set of maskable states $\mR_\mathrm{ND}$ does not live in a hyperdisk in $\mH_R$.

This example indicates that one can mask states which do not live in a single hyperdisk in $\mH_R$, even using the nondegenerate masking protocol.

\subsection{Completely Degenerate Case}
In this case, $\rho_A=\rho_B=\iden/d$. In contrast to the nondegenerate case, the set of legal states $\mL_{\mathrm{CD}}$ consists of all the maximally entangled states, and hence is not restricted to a single hyperdisk.
Precisely, $\mL_{\mathrm{CD}}$ is expressed as
\begin{equation}\label{eq:me}
    \ket{\Psi(U)}=\frac{1}{\sqrt{d}}\sum^{d-1}_{j=0}U\otimes\iden\ket{jj}=U\otimes \iden\ket{\Phi_\iden},\ U\in\mathcal{U}_d,
\end{equation}
where $\ket{\Phi_\iden}=\frac{1}{\sqrt{d}}\sum_j\ket{jj}$, $\mathcal{U}_d$ is the set of $d$-dimensional unitary operators, and $\{\ket{j}\}_{j=0}^{n-1}$ is an orthonormal basis of $\mH_{A/B}$.
The set of target states $\mT_\mathrm{CD}$ is then expressed as
\begin{equation}
    \mT_\mathrm{CD}=\mV_{\mT\mathrm{CD}}\cap\mL_{\mathrm{CD}}.
\end{equation}

It is worth noting that $\mathcal{V}_{\mL\mathrm{CD}}=\spanned\left(\mL_\mathrm{CD}\right)=\mH_{AB}$.
This is because the generalized Bell states $\ket{\Psi_{jk}}\equiv X^{j}Z^{k}\otimes\iden\ket{\Psi_\iden}$ constitute a complete orthogonal basis of $\mH_{AB}$ \cite{gbs}. Here $Z=\sum_k \exp\left(2k\pi i/d\right)\ket{k}\bra{k}$ and  $X=\sum_k\ket{(k+1)\mathrm{mod}\  d}\bra{k}$ are generalized Pauli operators.
Hence the dimension $n$ of input space is bounded as
\begin{equation}
    n=\mathrm{dim}\left(\mV_{\mT\mathrm{CD}}\right)\leq \mathrm{dim}\left(\mV_{\mL\mathrm{CD}}\right)=d^2.
\end{equation}
This means that in some situations the dimension of input space can be larger than the Schmidt number of the target states.

Instead of a hyperdisk, $\mL_{CD}$ is the set of maximally entangled states, so $\mT_\mathrm{CD}$ may not be a subset of any hyperdisk.
The following example shows that $\mT_\mathrm{CD}$ can consist of an infinite number of hyperdisks.
Here we set $n=3,d=2$ and $\mV_{\mT\mathrm{CD}}=\spanned\{\ket{00},\ket{11},\ket{01}+\ket{10}\}$. The set of target states reads $\mT_{CD}=\bigcup_{\xi\eta}\mS_{\xi\eta}$, where $\mS_{\xi\eta}$ are two-dimensional hyperdisks:
\begin{equation}\label{eq:37}
    \ket{\Psi_{\xi\eta}(\theta)}=\ket{\phi^+_{\xi\eta}\phi^+_{\xi\eta}}+e^{i\theta}\ket{\phi^-_{\xi\eta}\phi^-_{\xi\eta}}
\end{equation}
with
\begin{equation}
    \begin{aligned}
    \ket{\phi^{+}_{\xi\eta}}&=\cos\frac{\xi}{2}\ket{0}+\sin\frac{\xi}{2}e^{i\eta}\ket{1},\\
    \ket{\phi^{-}_{\xi\eta}}&=\sin\frac{\xi}{2}\ket{0}-\cos\frac{\xi}{2}e^{i\eta}\ket{1}.
    \end{aligned}
\end{equation}
Here $(\xi,\eta)$ is continuously chosen in $\mathbb{R}^2$.
Then we define masking machine $V_{\mathrm{mask}}$ as $\ket{0}\rightarrow\ket{00}$, $\ket{1}\rightarrow\ket{11}$, and $\ket{2}\rightarrow\frac{1}{\sqrt{2}}(\ket{01}+\ket{10})$.
The corresponding hyperdisks $\mS^\prime_{\xi\eta} = V^\dagger_\mathrm{mask}\mS_{\xi\eta}$ in the of set maskable states $\mR_{CD}$ are expressed as
\begin{align}
    \begin{aligned}[b]
    \ket{\psi_{\xi\eta}(\theta)}&=\cos^2\frac{\xi}{2}\ket{0}+\sin^2\frac{\xi}{2}e^{i2\eta}\ket{1}+\frac{1}{\sqrt{2}}\sin\xi e^{i\eta}\ket{2}\\
    &+e^{i\theta}\left(\sin^2\frac{\xi}{2}\ket{0}+\cos^2\frac{\xi}{2}e^{i2\eta}\ket{1}\right.\\
    &-\left.\frac{1}{\sqrt{2}}\sin\xi e^{i\eta}\ket{2}\right).\\
    \end{aligned}\raisetag{20pt}
\end{align}

Using Proposition \ref{pp:n}, we find these two-dimensional hyperdisks $\mS^\prime_{\xi\eta}$ are not subsets of a single hyperdisk in $\mH_R$.
Furthermore, the set of maskable states $\mR_{\mathrm{CD}}$ contains an unlimited number of hyperdisks, because of the continuity of the index $(\xi,\eta)$.
This example shows that, by using a completely degenerate masking machine, the number of hyperdisks contained in the set of maskable states can go to infinity. In this sense, we say that the degeneracy of the masking machine may enhance its power. 

\subsection{General Case}\label{sec:gmp}
In general, the marginal states $\rho_A$ and $\rho_B$ are partially degenerate.
The $j$th eigenspace of $\rho_A$ ($\rho_B$) is denoted as $\mH^{(j)}_A$ ($\mH^{(j)}_B$). Its eigenvalue, degeneracy, and basis are labeled as $\lambda_j$, $g(j)$, and $\left\{\ket{j,k}\right\}^{g(j)-1}_{k=0}$, respectively.
Then the legal states can be expressed as
\begin{equation}\label{eq:gt}
    \begin{aligned}
    \ket{\Psi(U)}&=\sum^{t-1}_{j=0}\sqrt{\lambda_j}\sum^{g(j)-1}_{k=0}U^{(j)}\otimes\iden\ket{j,k}\ket{j,k}\\
    &= U\otimes\iden\ket{\Psi_\iden},
    \end{aligned}
\end{equation}
where $\ket{\Psi_\iden}=\sum^{t-1}_{j=0}\sqrt{\lambda_j}\sum^{g(j)-1}_{k=0}\ket{j,k}\ket{j,k}$, $\sum^{t-1}_{j=0}g(j)=d$, and $U$ is a block-diagonal unitary
\begin{equation}\label{eq:bd}
    U=\bigoplus^{t-1}_{j=0} U^{(j)}
\end{equation}
with each block $U^{(j)}$ acting on $\mH^{(j)}_A$. The dimension $n$ of input space is then bounded as
\begin{equation}
    n=\dim\left(\mV_\mT\right)\leq \dim\left(\mV_\mL\right)=\sum^{t-1}_{j=0} g^2(j).
\end{equation}

In general, the whole set of legal states $\mL$ does not live in a single hyperdisk. The following lemma gives a necessary and sufficient condition that a subset of $\mL$ lives in a Schmidt hyperdisk. In other words, this lemma provides a criterion to determine whether $\mT$ is a subset of some Schmidt hyperdisk. 

\begin{lemma}\label{lm:ut}
    A set of states $\{\ket{\Psi(U)}\}_{U\in\mathcal{U}}$, with  $\ket{\Psi(U)}$ in the form of Eq. (\ref{eq:gt}), lives in a Schmidt hyperdisk $\mS^{AB}$, if and only if there exists a block-diagonal unitary matrix $U_T$ in the form of Eq. (\ref{eq:bd}) such that
    $\left[U U_T, U^\prime U_T\right]=0,\ \forall U,U^\prime\in\mathcal{U}$.
\end{lemma}

\begin{proof}
    First, we notice that $\ket{\Psi_\iden}$ can be reformulated as
    \begin{equation}\label{eq:alt}
        \ket{\Psi_\iden}=\sum^{t-1}_{j=0}\sqrt{\lambda_j}\sum^{g(j)-1}_{k=0}\ket{\phi^*_{jk}\phi_{jk}},
    \end{equation}
    where $\{\ket{\phi_{jk}}\}_k$ is an arbitrary orthonormal basis of $\mH^{(j)}_B$, and
   $\ket{\phi^*_{jk}}$ denotes the conjugate state of $\ket{\phi_{jk}}$, i.e.,  $\bra{j^\prime,k^\prime}\ket{\phi^*_{jk}}=\bra{\phi_{jk}}\ket{j^\prime,k^\prime},\ \forall i',j'$.

     \noindent \textbf{[Sufficient part]}: we start from the condition 
    \begin{equation}\label{eq:suf_con}
        \{\ket{\Psi(U)}\}_{U\in\mathcal{U}}\subseteq\mS^{AB}.
    \end{equation}
    Because the states $\ket{\Psi(U)}$ share the same Schmidt coefficients, any Schmidt hyperdisk $\mS^{AB}$ satisfying Eq. (\ref{eq:suf_con}) is in the following form
    \begin{equation}
        \ket{\Psi(\bm{\theta})}=\sum_j\sqrt{\lambda_j}\sum_k e^{i\theta_{jk}}\ket{\psi_{jk}\phi_{jk}},\ \bm{\theta}\in\mathbb{R}^d,
    \end{equation}
    where $\{\ket{\psi_{jk}}\}_k$ and $\{\ket{\phi_{jk}}\}_k$ are orthonormal bases of $\mH^{(j)}_A$ and $\mH^{(j)}_B$, respectively.
    Comparing with Eq. (\ref{eq:alt}), we find $\ket{\Psi(\bm{\theta})}=U(\bm{\theta})\otimes\iden\ket{\Psi_\iden}$ with
\begin{equation*}
U(\bm{\theta})=\bigoplus_j\left[\sum_k e^{i\theta_{jk}}\ket{\psi_{jk}}\bra{\phi^*_{jk}}\right].
\end{equation*}
It follows that there exists a block-diagonal unitary operator
\begin{equation*}
    U_T=\bigoplus_j\left[\sum_k \ket{\phi^*_{jk}}\bra{\psi_{jk}}\right],
\end{equation*}
    such that $\left[U(\bm{\theta}) U_T,U(\bm{\theta}^\prime) U_T\right]=0,\ \forall \bm{\theta},\bm{\theta'}\in\mathbb{R}^d$. From Eq. (\ref{eq:suf_con}), we have $\mathcal{U}\subseteq\{U(\bm{\theta})|\bm{\theta}\in\mathbb{R}^d\}$and, hence, $\left[U U_T,U' U_T\right]=0,\ \forall U,U'\in\mathcal{U}$.

    \noindent \textbf{[Necessary part]}: we start from the condition that there exists a block-diagonal unitary operator $U_T=\bigoplus_j U^{(j)}_T$ such that $\left[U U_T, U^\prime U_T\right]=0,\ \forall U,U'\in\mathcal{U}$.
    From the commutative property, the unitary operators $UU_T$ can be simultaneously diagonalized, i.e., there is an orthonormal basis $\{\ket{\psi_{jk}}\}_{jk}$ such that $\forall U\in\mathcal{U}$, and we have
    \begin{equation*}
        UU_T=\bigoplus_j\left[\sum_k e^{i\theta_{jk}}\ket{\psi_{jk}}\bra{\psi_{jk}}\right],
    \end{equation*}
    It follows that any unitary operator $U\in\mathcal{U}$ can be written as
    \begin{equation*}
        U=UU_T U_T^\dagger=\bigoplus_j\left[\sum_k e^{i\theta_{jk}}\ket{\psi_{jk}}\bra{\psi_{jk}}U_{T}^{(j)\dagger}\right]
    \end{equation*}
    Together with Eq. (\ref{eq:alt}), we arrive at 
    \begin{equation*}
         \ket{\Psi(U)}=U\otimes\iden\ket{\Psi_\iden}=\sum_j\sqrt{\lambda_j}\sum_k e^{i\theta_{jk}}\ket{\psi_{jk}\phi_{jk}},
    \end{equation*}
    where $\ket{\phi_{jk}}=\left(U_T^{(j)}\ket{\psi_{jk}}\right)^*$.
   This means that $\left\{\ket{\Psi(U)}\right\}_{U\in\mathcal{U}}$ lives in a Schmidt hyperdisk with hyperdisk basis $\left\{\ket{\psi_{jk}\phi_{jk}}\right\}_{jk}$.
\end{proof}

\section{Structure of the set of maskable states for qubits and qutrits}
In this section, we derive the explicit structures of the set of maskable states $\mR$ for two-dimensional and three-dimensional spaces. Note that instead of checking whether the maskable states lie on a hyperdisk we now focus on fully characterizing the structure of the set of maskable states. Recalling Def. \ref{def:masking}, the dimension of input space and the Schmidt number of target states are denoted by $n$ and $d$, respectively. 

\subsection{\texorpdfstring{$n=2,d\geq 2$}{n=2,d>=2}}
Here $\mH_R$ is restricted to be a qubit space while the dimension of $\mV_\mL$ is not limited.
The following theorem shows the structure of the set of maskable qubit states. 

\begin{theorem}\label{th:1}
    Let $\mH_R$ be a Hilbert space of qubits. The set of maskable states $\mR\subset\mH_R$ is either a two-dimensional hyperdisk or a set of two states.
\end{theorem}
\begin{proof}
    Because $\mR$ is isomorphic to the set of target states $\mT$, here we only need to prove that $\mT$ is either a two-dimensional hyperdisk or a set of two states.
    From $n=2$, there are at least two states, labeled as $\ket{\Psi_0}$ and $\ket{\Psi_1}$, in $\mT$, which in turn belongs to the set of legal states $\mL$. 
    In general, states in $\mL$ can be written in the form of Eq. (\ref{eq:gt}), so we have
    \begin{equation}
        \ket{\Psi_0}=U_0\otimes\iden\ket{\Psi_\iden},\quad 
        \ket{\Psi_1}=U_1\otimes\iden\ket{\Psi_\iden},
    \end{equation}
    where $\ket{\Psi_\iden}$ is an entangled state with Schmidt number $d\geq 2$, and $U_0$ and $U_1$ are two block-diagonal $d$-dimensional unitary matrices. 
    
    Because $\dim\left(\mV_\mT\right)=n=2$ and $\ket{\Psi_0},\ket{\Psi_1}\in\mT\subset\mV_\mT$, we have $\mV_\mT=\spanned\left\{\ket{\Psi_0},\ket{\Psi_1}\right\}$. 
    According to Eq. (\ref{eq:tl}), any target state $\ket{\Psi}\in\mT$ can be written as
    \begin{equation}
        \ket{\Psi}=a\ket{\Psi_0}+b\ket{\Psi_1}=U(a,b)\otimes\iden\ket{\Psi_\iden},
    \end{equation}
    where $a$ and $b$ are chosen such that $\ket{\Psi}$ is normalized and $U(a,b)$ is a unitary operator. Here the first equation is from $\ket{\Psi}\in\mV_\mT$ and the second equation is from $\ket{\Psi}\in\mL$. It follows that $U(a,b)=a U_0 + b U_1$.
    By choosing $U_T=U_0^\dagger$, we have $[U(a,b)U_T,U(a^\prime,b^\prime)U_T]=0$ for all $U(a,b)$.
    From Lemma \ref{lm:ut}, the set of target states live in a Schmidt hyperdisk $\mS^{AB}$, i.e. $\mT\subseteq\mS^{AB}$. 

    Therefore, $\mT=\mV_\mT\cap\mS^{AB}$. This means that $\mT$ is a two-dimensional regular subset of $\mS^{AB}$.
    From Lemma \ref{lm:c2}, $\mT$ consists of either a two-dimensional hyperdisk or two single states.
\end{proof}

By condition (2) in Def. \ref{def:masking}, for a given masking machine $V_\mathrm{mask}$, the set of maskable states contains all of the states which can be masked by $V_\mathrm{mask}$. Therefore, Theorem \ref{th:1} indicates that for qubit states as input, other than the masking machines the maskable states of which constitute a two-dimensional hyperdisk, there are masking machines for which the whole set of the machine's maskable states contains only two states. 

Since any two states belong to a two-dimensional hyperdisk, Theorem \ref{th:1} implies that the hyperdisk conjecture in Ref. \cite{masking} holds for the qubit case, which is one of the main results in Ref. \cite{PhysRevA.100.030304}. Here we emphasize that the statement in Theorem $\ref{th:1}$ is stronger, in that it characterizes all valid structures of $\mR$ in qubit space (no matter which masking machine is employed and no matter how large the Schmidt number of target states is). For instance, any three states on a Bloch sphere live on a disk, but from Theorem \ref{th:1} there is no masking machine the maskable states of which contain only three states. In fact, for any masking machine which can mask more than two states, the set of its maskable states constitutes a two-dimensional hyperdisk.

\subsection{\texorpdfstring{$n=3,d=3$}{n=3,d=3}}
For $\mH_R$ with higher dimension, we have shown in the last section that the set of maskable states may not belong to a single hyperdisk.
The following theorem provides a series of explicit structures of the set of target states $\mT$, and the structure of the set of maskable states $\mR$ is inferred from the isometry \hyperref[pty:2]{Property 2}.

\begin{theorem}\label{lmgcn3}
    For $n=d=3$, if the set of target states $\mT$ contains at least one two-dimensional subhyperdisk of a Schmidt hyperdisk $\mS$, then the structure of $\mT$ is one of the following three types.
        \item type I: $\mT$ is a three-dimensional Schmidt hyperdisk;
        \item type II: $\mT$ consists of two two-dimensional subhyperdisks locating on two different Schmidt hyperdisks;
        \item type III: $\mT$ consists of a two-dimensional subhyperdisk of a Schmidt hyperdisk and a single state locating on another Schmidt hyperdisk.
\end{theorem}

For a nondegenerate masking protocol with $n=d$, the set of target states must have type I form, which has been discussed below Eq. (\ref{ineq}). The complete proof of Theorem \ref{lmgcn3} will be given in Appendices \ref{app:A} and \ref{app:B}. Here we show the general expression of each type of $\mT$.\\
\emph{type I:} $\mT=\{ \ket{\Psi(\theta_1,\theta_2)}|\theta_1,\theta_2\in[0,2\pi)\}$, where
\begin{equation}
        \Scale[0.9]{\ket{\Psi(\theta_1,\theta_2)}=\sqrt{\lambda_0}\ket{00}+e^{i\theta_1}\sqrt{\lambda_1}\ket{11}+e^{i\theta_2}\sqrt{\lambda_2}\ket{22}.}
\end{equation}
Here $\left\{\ket{00},\ket{11},\ket{22}\right\}$ is a Schmidt basis.\\
\emph{type II:}  $\mT=\left\{\ket{\Psi_0(\alpha)},\ket{\Psi_1(\beta)}|\alpha,\beta\in[0,2\pi)\right\}$, where
\begin{equation}\label{eq:type2}
    \Scale[0.9]{\begin{aligned}
    \ket{\Psi_0(\alpha)}&=\sqrt{\lambda_1}\ket{00}+e^{i\alpha}\left(\sqrt{\lambda_1}\ket{11}+\sqrt{\lambda_2}\ket{22}\right),\\
    \ket{\Psi_1(\beta)}&=\sqrt{\lambda_1}\left(\ket{\phi^-_{01}\psi^-_{01}}+e^{i\beta}\ket{\phi^+_{01}\psi^+_{01}}\right)+\sqrt{\lambda_2}\ket{22}.
    \end{aligned}}
\end{equation}
Here $\left\{\ket{\phi^+_{01}},\ket{\phi^-_{01}}\right\}$ and $\left\{\ket{\psi^+_{01}},\ket{\psi^-_{01}}\right\}$ are two orthogonal bases of $\spanned\left\{\ket{0},\ket{1}\right\}$ and satisfy $0<\left|\bra{0}\ket{\phi^+_{01}}\right|=\left|\bra{0}\ket{\psi^+_{01}}\right|<1$.\\
\emph{type III:}  $\mT=\{\ket{\Psi_0(\alpha)}, \ket{\Psi^\prime}|\alpha\in[0,2\pi)\}$, where
\begin{equation}
    \begin{aligned}
    &\quad\ket{\Psi_0(\alpha)}=\ket{00}+e^{i\alpha}\left(\ket{11}+\ket{22}\right),\\
        \ket{\Psi^\prime}&=\cos\frac{\theta}{2}\ket{00}+\sin\frac{\theta}{2}\left(e^{i\varphi_0}\ket{10}+e^{i\varphi_1}\ket{0}\ket{\psi^+_{12}}\right)\\
        &+e^{i(\varphi_0+\varphi_1)}\left(e^{i\eta}\ket{2}\ket{\psi^-_{12}}-\cos\frac{\theta}{2}\ket{1}\ket{\psi^+_{12}}\right).
    \end{aligned}
\end{equation}
Here $\theta\in(0,\pi]$, $\left|\bra{1}\ket{\psi^+_{12}}\right|\neq 1$, $\{\ket{\psi^+_{12}},\ket{\psi^-_{12}}\}$ is an orthonormal basis of $\spanned\{\ket{1},\ket{2}\}$, and $\eta,\varphi_0$ and $\varphi_1$ are relative phases in $[0,2\pi)$.

From the above expressions, we can see that completely degenerate masking machines can realize all  three types of target states, partially degenerate masking machines can realize types I and II, while a nondegenerate masking machine can only realize type I.

It is worth mentioning that, if we neglect the condition that $\mT$ contains at least one two-dimensional subhyperdisk of the Schmidt hyperdisk,
then $\mT$ has structures other than the above three types. A simple example is
\begin{equation}\label{eq:spc1}
    \mT=\left\{\ket{\Phi_\iden},\ Z\otimes\iden\ket{\Phi_\iden},\ X\otimes\iden\ket{\Phi_\iden}\right\},
\end{equation}
where $\ket{\Phi_\iden}=\frac{1}{\sqrt{3}}\left(\ket{00}+\ket{11}+\ket{22}\right)$, and $X$ and $Z$ are three-dimensional generalized Pauli matrices. In this example, $\mT$ only consists of three orthogonal states. This set of target states is obtained when we set $\mV_\mT=\spanned\left\{\ket{\Phi_\iden},\ Z\otimes\iden\ket{\Phi_\iden},\ X\otimes\iden\ket{\Phi_\iden}\right\}$ and $\mL$ to be the set of maximally entangled states.

Further, the following example shows another structure of $\mT$. When we set $\mV_\mT=\spanned\{\ket{00},\ket{11},\sqrt{\lambda_1/2}i(\ket{01}+\ket{10})+\sqrt{\lambda_2}\ket{22}\}$, and set $\mL$ to be the set of bipartite states which have partially degenerate marginal states $\rho_A=\rho_B=\lambda_1\left(\ket{0}\bra{0}+\ket{1}\bra{1}\right)+\lambda_2\ket{2}\bra{2}$, the set of target states reads
\begin{equation}\label{eq:spc2}
    \begin{aligned}
        \ket{\Psi(\eta)}&=\sqrt{\lambda_1/2}\left(e^{i\eta}\ket{00}+e^{-i\eta}\ket{11}\right)\\
        &+\sqrt{\lambda_1/2}i\left(\ket{01}+\ket{10}\right)+\sqrt{\lambda_2}\ket{22}.
    \end{aligned}
\end{equation}
As for this example, $\mT$ contains an infinite number of states, but does not contain any nontrivial hyperdisk (i.e., the dimension of the hyperdisk is strictly larger than 1). 

\section{conclusion}

We have studied the structure of the set of maskable states and its relation to hyperdisks. Precisely, we develop a general method to determine the set of target states (which is isomorphic to the set of maskable states) , and prove criteria for judging whether a set of states belongs to a hyperdisk.
we find that the structure of maskable states depends on the dimension $n$ of the input space, the Schmidt number $d$ of the target states, and the degeneracy of marginal states. Further, we derive the valid structures of the set of maskable states for the two cases with $n=2,d\geq 2$ and with $n=d=3$. In doing so, we prove the hyperdisk conjecture in Ref. \cite{masking} for $n=2$, and disprove it for $n>2$.

In most of the cases we have considered in this paper, the set of maskable state consists of a finite amount of hyperdisks. However, when the degeneracy of marginal states goes high, masking machines can be designed to mask an infinite number of hyperdisks [see Eq. (\ref{eq:37}) as an example]. This is an evidence that degenerate masking machines are more powerful than the nondegenerate ones. Nevertheless, it is an open question for a nondegenerate masking machine  whether the number of hyperdisks in the set of maskable states is always finite. A related open question is whether a regular subset of a hyperdisk can be covered with finite number of hyperdisks (also see discussion above Lemma \ref{lm:c2}).


\begin{acknowledgments}
    This work was supported by National Natural Science Foundation of China under Grant No. 11774205, and the Young Scholars Program of Shandong University.
\end{acknowledgments}

\appendix
\onecolumngrid

\section{Proof for the structure of partially degenerate masker when \texorpdfstring{$n=d=3$}{n=d=3}}\label{app:A}
    
    In an $n=d=3$ partially degenerate masking protocol, without loss of generality, we write down the marginal states as
    \begin{equation}\label{eq:margin}
        \rho_A=\rho_B=\lambda_1\ket{0}\bra{0}+\lambda_1\ket{1}\bra{1}+\lambda_2\ket{2}\bra{2},
    \end{equation}
    where $\lambda_1\neq\lambda_2$. 
    Thus, $\mL$ is fixed as
    \begin{equation}
        \ket{\Psi(U)}=U\otimes\iden\left(\sqrt{\lambda_1}\ket{00}+\sqrt{\lambda_1}\ket{11}+\sqrt{\lambda_2}\ket{22}\right), U\in\mathcal{U}.
    \end{equation}
    Here $\mathcal{U}$ is the set of block-diagonal unitary matrices in the form $U=U_{01}+e^{i\eta}\ket{2}\bra{2}$, where $U_{01}$ is an arbitrary unitary matrix acting on $\spanned\left\{\ket{0},\ket{1}\right\}$ and $\eta\in[0,2\pi)$.
    Recalling $\mT=\mL\cap\mV_\mT$, we find that $\mT$ is fully characterized by $\mV_\mT$.
    Thus, we mainly focus on deriving $\mV_\mT$ to get the different type of $\mT$.

    From the condition of Theorem 2, there is a two-dimensional subhyperdisk $\mS_0$ of Schmidt hyperdisk $\mS$. 
    According to the general form of subhyperdisks, there are two possible forms of $\mS_0$, which are expressed as
    \begin{subequations}\label{eq:12choice}
    \begin{align}
    \ket{\Psi_0(\alpha)}&=\sqrt{\lambda_1}\left(\ket{00}+\ket{11}\right)+e^{i\alpha}\sqrt{\lambda_2}\ket{22},\label{eq:1choice}\\
    \ket{\Psi_0(\alpha)}&=\sqrt{\lambda_1}\ket{00}+e^{i\alpha}\left(\sqrt{\lambda_1}\ket{11}+\sqrt{\lambda_2}\ket{22}\right).\label{eq:2choice}
    \end{align}
    \end{subequations}
    If $\mS_0$ takes the form in Eq. (\ref{eq:1choice}), then $\mT$ takes the type I form. In order to prove this, we need to derive $\mV_\mT$. First, an arbitrary state $\ket{\Psi^\prime}\in\mT\setminus\mS_0$ is a legal state, and thus can be written as
    \begin{equation}
    \begin{aligned}
        \ket{\Psi^\prime}&=U\otimes\iden\left(\sqrt{\lambda_1}\ket{00}+\sqrt{\lambda_1}\ket{11}+\sqrt{\lambda_2}\ket{22}\right)\\
        &=\left(\ket{\phi^-_{01}}\bra{\phi^-_{01}}+e^{i\gamma_0}\ket{\phi^+_{01}}\bra{\phi^+_{01}}+e^{i\gamma_1}\ket{2}\bra{2}\right)\otimes\iden\left(\sqrt{\lambda_1}\ket{00}+\sqrt{\lambda_1}\ket{11}+\sqrt{\lambda_2}\ket{22}\right)\\
        &=\sqrt{\lambda_1}\left(\ket{\phi^-_{01}\phi^{-*}_{01}}+e^{i\gamma_0}\ket{\phi^+_{01}\phi^{+*}_{01}}\right)+e^{i\gamma_1}\sqrt{\lambda_2}\ket{22},
    \end{aligned}
    \end{equation}
    where $\left\{\ket{\phi^-_{01}},\ket{\phi^+_{01}}\right\}$ is an arbitrary orthonormal basis for subspace $\spanned\{\ket{0},\ket{1}\}$ and $\gamma_0\neq0$.
    Here the second equation is due to the fact that the unitary $U_{01}$ can generally be written as $U_{01}=\ket{\phi^-_{01}}\bra{\phi^-_{01}}+e^{i\gamma_0}\ket{\phi^+_{01}}\bra{\phi^+_{01}}+e^{i\gamma_1}\ket{2}\bra{2}$.
    The condition $\gamma_0\neq0$ makes sure that $\ket{\Psi'}\notin\mS_0$.
    Further, $\mS_0$ in the form of Eq. (\ref{eq:1choice}) can be reformulated as
    \begin{equation}
        \ket{\Psi_0(\alpha)}=\sqrt{\lambda_1}\left(\ket{\phi^-_{01}\phi^{-*}_{01}}+\ket{\phi^+_{01}\phi^{+*}_{01}}\right)+e^{i\alpha}\sqrt{\lambda_2}\ket{22}.
    \end{equation}
    To sum up, we arrive at $\mV_\mT=\spanned\{\mS_0,\ket{\Psi'}\}=\spanned\left\{\ket{\phi^-_{01}\phi^{-*}_{01}},\ket{\phi^+_{01}\phi^{+*}_{01}},\ket{22}\right\}$. Thus, we find that $\mT=\mL\cap\mV_\mT$ takes the following form:
    \begin{equation}
        \ket{\Psi(\eta_0,\eta_1)}=\sqrt{\lambda_1}\left(\ket{\phi^-_{01}\phi^{-*}_{01}}+e^{i\eta_0}\ket{\phi^+_{01}\phi^{+*}_{01}}\right)+e^{i\eta_1}\sqrt{\lambda_2}\ket{22}.
    \end{equation}
    Namely, $\mT$ is a Schmidt hyperdisk in this situation. This is type I in Theorem 2.
    
    If $\mS_0$ takes the form in Eq. (\ref{eq:2choice}), then $\mT$ can take the type II form. The proof is sketched as follows.

\begin{itemize}
\item \textbf{Part 1}. In this part, we show that, there can exist two two-dimensional subhyperdisks of different Schmidt hyperdisks in $\mT$.
\item \textbf{Part 2}. In this part, we prove that, if there are two two-dimensional subhyperdisks of different Schmidt hyperdisks in $\mT$, then $\mT$ does not contain a third two-dimensional subhyperdisk. 
\item \textbf{Part 3}. In this part, we prove that, besides type I and type II, $\mT$ does not take other types. 
\end{itemize}
    
    \noindent [\textbf{Part 1}]:  To prove that there can exist two two-dimensional subhyperdisks of different Schmidt hyperdisks in $\mT$, we just need to check that the two hyperdisks as in Eq. (\ref{eq:type2}) can be contained in $\mT$. Here we denote the two hyperdisk  $\{\ket{\Psi_0(\alpha)}\}_\alpha$ and $\{\ket{\Psi_1(\beta)}\}_\beta$ in Eq. (\ref{eq:type2}) as $\mS_0^\prime$ and $\mS_1^\prime$, respectively. The corresponding hyperdisk bases $\mB^\prime_0$ and $\mB^\prime_1$ are
    \begin{equation}
    \begin{aligned}
        \mB^\prime_0&=\left\{\ket{\Phi^0_0}:=\ket{00},\ket{\Phi^0_1}:=\frac{1}{\sqrt{\lambda_1+\lambda_2}}\left(\sqrt{\lambda_1}\ket{11}+\sqrt{\lambda_2}\ket{22}\right)\right\},\\
        \mB^\prime_1&=\left\{\ket{\Phi^1_0}:=\ket{\phi^+_{01}\psi^+_{01}},\ket{\Phi^1_1}:=\frac{1}{\sqrt{\lambda_1+\lambda_2}}\left(\sqrt{\lambda_1}\ket{\phi^-_{01}\psi^-_{01}}+\sqrt{\lambda_2}\ket{22}\right)\right\},
    \end{aligned}
    \end{equation}
    where $\{\ket{\phi^+_{01}},\ket{\phi^-_{01}}\}$ and $\{\ket{\psi^+_{01}},\ket{\psi^-_{01}}\}$ are orthogonal bases of $\spanned\{\ket{0},\ket{1}\}$ and satisfy $\left.0<\left|\bra{0}\ket{\phi^+_{01}}\right|=\left|\bra{0}\ket{\psi^+_{01}}\right|<1\right.$.

    Then we prove the following statements.
    \begin{enumerate}
        \item $\mS^\prime_0$ and $\mS^\prime_1$ are contained in $\mL$; 
        \item $\mS^\prime_0$ and $\mS^\prime_1$ are contained in different Schmidt hyperdisks;
        \item The dimension of $\spanned\{\mS^\prime_0\cup\mS^\prime_1\}=\spanned\{\mB^\prime_0\cup\mB^\prime_1\}$ is equal to 3, such that $\dim(\mT)\equiv n=3$.
    \end{enumerate}
    
    \emph{Proof of 1}. The states in both $\mS^\prime_0$ and $\mS^\prime_1$ have marginal states in the form of Eq. (\ref{eq:margin}), so we have $\mS'_0\cup\mS'_1\subset \mL$.
    
    \emph{Proof of 2}. By Proposition \ref{pp:n}, $\mS_0^\prime$ and $\mS_1^\prime$ do not belong to a single three-dimensional hyperdisk. Further, $\mS_0^\prime$ is a subset to a Schmidt hyperdisk with basis $\{\ket{00},\ket{11},\ket{22}\}$, while $\mS_1^\prime$ is a subset to a Schmidt hyperdisk with basis $\{\ket{\phi^+_{01}\psi^+_{01}},\ket{\phi^-_{01}\psi^-_{01}},\ket{22}\}$.
    
    \emph{Proof of 3}. We check that only three states in $\mB^\prime_0\cup\mB^\prime_1$ are linearly independent, so we get $\dim(\mV^\prime)=3$. 
    
    To sum up, two two-dimensional subhyperdisks of different Schmidt hyperdisks can be contained in $\mT$.

    \noindent [\textbf{Part 2}]: Next, we suppose there are two two-dimensional subhyperdisks of different Schmidt hyperdisks in $\mT$. 
    Our main goal is then to prove there does not exist any other two-dimensional subhyperdisk in $\mT$. 
    
    As we have fixed $\mS_0$ in Eq. (\ref{eq:2choice}), we mainly focus on the second subhyperdisk $\mS_1$ of another Schmidt hyperdisk $\mS^\prime$. The Schmidt hyperdisk $\mS^\prime$ is written as
    \begin{equation}
        \begin{aligned}
        \ket{\Psi^\prime(\bm{\eta})}=\sqrt{\lambda_1}&\left(\ket{\phi^-_{01}\psi^-_{01}}+e^{i\eta_0}\ket{\phi^+_{01}\psi^+_{01}}\right)+e^{i\eta_1}\sqrt{\lambda_2}\ket{22},\\
        \ket{\phi^+_{01}}=\cos\frac{\theta_0}{2}\ket{0}+&\sin\frac{\theta_0}{2}e^{i\varphi_0}\ket{1},\quad\ket{\phi^-_{01}}=\sin\frac{\theta_0}{2}e^{-i\varphi_0}\ket{0}-\cos\frac{\theta_0}{2}\ket{1},\\
        \ket{\psi^+_{01}}=\cos\frac{\theta_1}{2}\ket{0}+&\sin\frac{\theta_1}{2}e^{i\varphi_1}\ket{1},\quad\ket{\psi^-_{01}}=\sin\frac{\theta_1}{2}e^{-i\varphi_1}\ket{0}-\cos\frac{\theta_1}{2}\ket{1}.
        \end{aligned}
    \end{equation}
    The parameters $\eta_0,\eta_1,\varphi_0$, and $\varphi_1$ have the domain $[0,2\pi)$. 
    The pair of parameters $(\theta_0,\theta_1)$ has the domain $[0,\pi]\times[0,\pi]\setminus\left\{(0,0),(\pi,\pi)\right\}$,
    because the condition $\mS\neq\mS^\prime$ forces $(\theta_0,\theta_1)\neq (0,0) \textnormal{ or } (\pi,\pi)$.
    In a similar way to Eq. (\ref{eq:12choice}), there are two possible situations of $\mS_1$:
    \begin{subequations}
        \begin{align}
        \ket{\Psi_1(\beta)}&=\sqrt{\lambda_1}\left(\ket{\phi^-_{01}\psi^-_{01}}+e^{i\eta}\ket{\phi^+_{01}\psi^+_{01}}\right)+\sqrt{\lambda_2}e^{i\beta}\ket{22},\label{eq:sub1choice}\\
        \ket{\Psi_1(\beta)}&=\sqrt{\lambda_1}\left(\ket{\phi^-_{01}\psi^-_{01}}+e^{i\beta}\ket{\phi^+_{01}\psi^+_{01}}\right)+\sqrt{\lambda_2}e^{i\eta}\ket{22},\label{eq:sub2choice}
        \end{align}
    \end{subequations}
    where $\eta$ is a constant real number. The first situation Eq. (\ref{eq:sub1choice}) follows $\mathrm{dim}(\mV_\mT)>3$, which is a contradiction to $n=3$. 
    Thus, we choose the second situation Eq. (\ref{eq:sub2choice}) as $\mS_1$ in the following. 
    
    Let us write down the hyperdisk basis of Eq. (\ref{eq:sub2choice}):
    \begin{equation}
            \left\{\ket{\Phi_0} := \ket{\phi^+_{01}\psi^+_{01}},\quad \ket{\Phi_1} :=\frac{1}{\sqrt{\lambda_1+\lambda_2}}\left(\sqrt{\lambda_1}\ket{\phi^-_{01}\psi^-_{01}}+\sqrt{\lambda_2}e^{i\eta}\ket{22}\right)\right\}.
    \end{equation}
    Then $\mV_\mT=\spanned\left\{\ket{00},\sqrt{\lambda_1}\ket{11}+\sqrt{\lambda_2}\ket{22},\ket{\Phi_0},\ket{\Phi_1}\right\}$. Furthermore, we define the orthogonal projections of $\ket{\Phi_0}$ and $\ket{\Phi_1}$ to  $\spanned\left\{\ket{00},\sqrt{\lambda_1}\ket{11}+\sqrt{\lambda_2}\ket{22}\right\}$:
    \begin{equation}
        \begin{aligned}
            \ket{\Phi^0_\perp} &= \ket{\Phi_0} - \bra{00}\ket{\Phi_0}\ket{00} - \frac{\sqrt{\lambda_1}\bra{11}+\sqrt{\lambda_2}\bra{22}}{\sqrt{\lambda_1+\lambda_2}}\ket{\Phi_0}\frac{\sqrt{\lambda_1}\ket{11}+\sqrt{\lambda_2}\ket{22}}{\sqrt{\lambda_1+\lambda_2}},\\
            \ket{\Phi^1_\perp} &= \ket{\Phi_1} - \bra{00}\ket{\Phi_1}\ket{00} - \frac{\sqrt{\lambda_1}\bra{11}+\sqrt{\lambda_2}\bra{22}}{\sqrt{\lambda_1+\lambda_2}}\ket{\Phi_1}\frac{\sqrt{\lambda_1}\ket{11}+\sqrt{\lambda_2}\ket{22}}{\sqrt{\lambda_1+\lambda_2}}.\\
        \end{aligned}
    \end{equation}
    Here we find that $\ket{\Phi^0_\perp}$ and $\ket{\Phi^1_\perp}$ must be collinear such that $\mathrm{dim}(\mV_\mT)=3$. 
    By using the collinearity, we derive $\theta_0=\theta_1=\theta$ and $\eta = 0$. Thus, $\mS_1$ can be parametrized by $(\theta,\varphi_0,\varphi_1)$, where $\theta\in(0,\pi)$. 
    Then both $\ket{\Phi^0_\perp}$ and $\ket{\Phi^1_\perp}$ are collinear to
    \begin{equation}\label{eq:balanced_basis}
        \ket{\Phi_\perp} = \sin\theta\sqrt{\lambda_1}\left(e^{-i\varphi_0}\ket{01}+e^{-i\varphi_1}\ket{10}\right)+\frac{\lambda_1\lambda_2}{\lambda_1+\lambda_2}(1-\cos\theta)\left(\frac{1}{\sqrt{\lambda_1}}\ket{11}-\frac{1}{\sqrt{\lambda_2}}\ket{22}\right).
    \end{equation}
    It follows that there is a one-to-one correspondence between $\ket{\Phi_\perp}$ and the tuple $(\theta,\varphi_0,\varphi_1)$. 
    Then we have $\mV_\mT=\spanned\left\{\ket{00},\sqrt{\lambda_1}\ket{11}+\sqrt{\lambda_2}\ket{22},\ket{\Phi_\perp}\right\}$, where $\ket{\Phi_\perp}$ is uniquely determined by the tuple $(\theta,\varphi_0,\varphi_1)$, and it in turn corresponds to $\mS_1$.
    Moreover, we find that $\mT$ contains at least two two-dimensional subhyperdisks of different Schmidt hyperdisks if and only if
    \begin{equation}\label{eq:mvt}\mV_\mT=\left\{\ket{00},\sqrt{\lambda_1}\ket{11}+\sqrt{\lambda_2}\ket{22},\ket{\Phi_\perp}\right\},\end{equation} 
    where $\ket{\Phi_\perp}$ is in Eq. (\ref{eq:balanced_basis}). 

    \noindent [\textbf{Part 3}]: In order to determine $\mV_\mT$ and get the full characterization of $\mT$, we choose an arbitrary state as $\ket{\Psi^\prime}\in\mL\setminus\mS_0$. If $\ket{\Psi^\prime}\in\mS$, it will lead to a type I structure of $\mT$. Otherwise, we have $\ket{\Psi^\prime}\in\mL\setminus\mS$, which can be parametrized as
    \begin{equation} \label{eq:unbalanced_exp}
        \ket{\Psi^\prime} = U_{(\theta,\varphi_0,\varphi_1,\eta)}\otimes\iden\left(\sqrt{\lambda_1}\ket{00}+\sqrt{\lambda_1}\ket{11}+\sqrt{\lambda_2}\ket{22}\right),
    \end{equation}
    where the block-diagonal local unitary matrix $U_{(\theta,\varphi_0,\varphi_1,\eta)}$ is written as
    \begin{equation} 
        \begin{aligned}
        U_{(\theta,\varphi_0,\varphi_1,\eta)}=\ket{\chi^+_{01}}\bra{0}+&e^{i\varphi_0}\ket{\chi^-_{01}}\bra{1}+e^{i\eta}\ket{2}\bra{2},\\
        \ket{\chi^+_{01}}=\cos\frac{\theta}{2}\ket{0}+\sin\frac{\theta}{2}e^{i\varphi_1}\ket{1} &,\quad \ket{\chi^-_{01}}=\sin\frac{\theta}{2}\ket{0}-\cos\frac{\theta}{2}e^{i\varphi_1}\ket{1}.
    \end{aligned}
    \end{equation}
    The phase parameters $\eta,\varphi_0$, and $\varphi_1$ have domain $[0,2\pi)$ while $\theta\in(0,\pi]$. 

    Next, define the orthogonal projection of $\ket{\Psi^\prime}$ to $\spanned\left\{\ket{00},\sqrt{\lambda_1}\ket{11}+\sqrt{\lambda_2}\ket{22}\right\}$
    \begin{equation}\label{eq:pw}
        \begin{aligned}
            \ket{\Psi^\prime_\perp}&=\ket{\Psi^\prime} - \bra{00}\ket{\Psi^\prime}\ket{00} - \frac{\sqrt{\lambda_1}\bra{11}+\sqrt{\lambda_2}\bra{22}}{\sqrt{\lambda_1+\lambda_2}}\ket{\Psi^\prime}\frac{\sqrt{\lambda_1}\ket{11}+\sqrt{\lambda_2}\ket{22}}{\sqrt{\lambda_1+\lambda_2}}\\
        &= \sqrt{\lambda_1}\sin\frac{\theta}{2}\left(e^{-i\varphi_0}\ket{01}+e^{-i\varphi_1}\ket{10}\right)-\frac{\lambda_1\lambda_2}{\lambda_1+\lambda_2}\left(\cos\frac{\theta}{2}+e^{i(\eta-\varphi_0-\varphi_1)}\right)\left(\frac{1}{\sqrt{\lambda_1}}\ket{11}-\frac{1}{\sqrt{\lambda_2}}\ket{22}\right).
        \end{aligned}
    \end{equation}
    Thus we have $\mV_\mT=\spanned\left\{\ket{00},\sqrt{\lambda_1}\ket{11}+\sqrt{\lambda_2}\ket{22},\ket{\Psi^\prime_\perp}\right\}$.
    However, the parametrization of Eq. (\ref{eq:pw}) has one degree of redundancy, since we can always find another state $\ket{\Psi^{\prime\prime}_\perp}$ which is collinear to $\ket{\Psi^\prime_\perp}$:
    \begin{equation}\label{eq:unbalanced_basis}
        \ket{\Psi^{\prime\prime}_\perp} = \sqrt{\lambda_1}\sin\frac{\theta^\prime}{2}\left(e^{-i\varphi^\prime_0}\ket{01}+e^{-i\varphi^\prime_1}\ket{10}\right)+\frac{\lambda_1\lambda_2}{\lambda_1+\lambda_2}\cos\frac{\theta^\prime}{2}\left(\frac{1}{\sqrt{\lambda_1}}\ket{11}-\frac{1}{\sqrt{\lambda_2}}\ket{22}\right),
    \end{equation}
    where $\theta^\prime\in(0,\pi)$ and $\varphi_0^\prime,\varphi_1^\prime\in(0,2\pi)$. 
    Notice that Eqs. (\ref{eq:balanced_basis}) and (\ref{eq:unbalanced_basis}) are essentially equivalent. 
    Namely, we arrive at $\mV_\mT=\left\{\ket{00},\sqrt{\lambda_1}\ket{11}+\sqrt{\lambda_2}\ket{22},\ket{\Psi^{\prime\prime}_\perp}\right\}$, which follows the same rule of Eq. (\ref{eq:mvt}).
    Thus, a second two-dimensional subhyperdisk of a different Schmidt hyperdisk  must be induced by $\ket{\Psi^\prime}$.
    Together with the conclusion of part 2 (i.e., a third two-dimensional subhyperdisk does not exist), we find that there does not exist a single state other than the two two-dimensional subhyperdisks.
    In other words, $\mT\setminus(\mS_0\cup\mS_1)=\emptyset$.
    Therefore, as long as we have fixed the first subhyperdisk $\mS_0$ of Schmidt hyperdisk $\mS$, every $\ket{\Psi^\prime}\in\mL\setminus\mS$ only leads $\mT$ to take the type II form. 
    
    In summary, there are only two types of $\mT$ as type I and type II in the partially degenerate case.

    \section{Proof for the structure of completely degenerate masker when \texorpdfstring{$n=d=3$}{n=d=3}}\label{app:B}

    In this section, we mainly discuss the type III structure. 
    Note that we may rearrange the phase parameters to simplify the expression in this section, 
    so notions such as $\varphi,\eta$ and $\omega$ sometimes do not directly represent the same phases in different formulas.
    In a similar way to Appendix \ref{app:A}, we also focus on deriving $\mV_\mT$ in order to get $\mT=\mL\cap\mV_\mT$ in the completely degenerate case, since $\mL$ is fixed as the set of three-dimensional maximally entangled states
    \begin{equation}
        \ket{\Psi(U)}=U\otimes\iden\left(\ket{00}+\ket{11}+\ket{22}\right),\quad \forall U\in\mathcal{U},
    \end{equation}
    where $\mathcal{U}$ is the set of three-dimensional unitary matrices acting on $\spanned\left\{\ket{0},\ket{1},\ket{2}\right\}$.
    
    First of all, we suppose there is a two-dimensional subhyperdisk $\mS_0$ of Schmidt hyperdisk $\mS$ in $\mT$. Without loss of generality,  $\mS_0$ is set as below in this section:
    \begin{equation}
        \ket{\Psi_0(\alpha)}=\ket{00}+e^{i\alpha}\left(\ket{11}+\ket{22}\right).
    \end{equation}

    Thus, the remaining degree of freedom is $\ket{\Psi^\prime}\in\mT\setminus\mS_0$.  Since the case that $\ket{\Psi^\prime}\in\mS$ leads $\mT$ to form a type I structure, we also set $\ket{\Psi^\prime}\in\mL\setminus\mS$ in the following context, in order to derive nontrivial situations.
    Then, $\ket{\Psi^\prime}$ can be parametrized as
    \begin{equation}\label{eq:mub}
        \begin{aligned}
        \ket{\Psi^\prime} =\ U_{(\theta,\nu_0,\nu_0,\varphi_0,\varphi_1,\omega_0,\omega_1,\eta)}\otimes\iden&\left(\ket{00}+\ket{11}+\ket{22}\right),\\
         U_{(\theta,\nu_0,\nu_0,\varphi_0,\varphi_1,\omega_0,\omega_1,\eta)} = \ket{\chi_0}\bra{0}+e^{i\varphi_1}&\ket{\chi_1}\bra{1} + e^{i(\varphi_1+\omega_1)}\ket{\chi_2}\bra{2},
        \end{aligned}
    \end{equation}
    where $\varphi_1,\omega_1\in[0,2\pi)$. $\left\{\ket{\chi_0},\ket{\chi_1},\ket{\chi_2}\right\}$ is an arbitrary orthogonal basis of $\spanned\left\{\ket{0},\ket{1},\ket{2}\right\}$
    \begin{equation}
        \begin{aligned}
            \ket{\chi_0}&=\cos\frac{\theta}{2}\ket{0}+e^{i\varphi_0}\sin\frac{\theta}{2}\cos\frac{\nu_0}{2}\ket{1}+e^{i(\varphi_0+\omega_0)}\sin\frac{\theta}{2}\sin\frac{\nu_0}{2}\ket{2},\\
            \ket{\chi_1}&=\sin\frac{\theta}{2}\cos\frac{\nu_1}{2}\ket{0}+e^{i\varphi_0}\left(e^{i\eta}\sin\frac{\nu_0}{2}\sin\frac{\nu_1}{2}-\cos\frac{\theta}{2}\cos\frac{\nu_0}{2}\cos\frac{\nu_1}{2}\right)\ket{1}\\
            &-e^{i(\varphi_0+\omega_0)}\left(e^{i\eta}\cos\frac{\nu_0}{2}\sin\frac{\nu_1}{2}+\cos\frac{\theta}{2}\sin\frac{\nu_0}{2}\cos\frac{\nu_1}{2}\right)\ket{2},\\
            \ket{\chi_2}&=\sin\frac{\theta}{2}\sin\frac{\nu_1}{2}\ket{0}-e^{i\varphi_0}\left(e^{i\eta}\sin\frac{\nu_0}{2}\cos\frac{\nu_1}{2}+\cos\frac{\theta}{2}\cos\frac{\nu_0}{2}\sin\frac{\nu_1}{2}\right)\ket{1}\\
            &+e^{i(\varphi_0+\omega_0)}\left(e^{i\eta}\cos\frac{\nu_0}{2}\cos\frac{\nu_1}{2}-\cos\frac{\theta}{2}\sin\frac{\nu_0}{2}\sin\frac{\nu_1}{2}\right)\ket{2},
        \end{aligned}
    \end{equation}
    where $\theta\in(0,\pi]$, $\nu_0,\nu_1\in[0,\pi]$ and $\varphi_0,\omega_0,\in[0,2\pi)$.
    Furthermore, we expand Eq. (\ref{eq:mub}) as below:
    \begin{equation}\label{eq:appcomplex}
        \begin{aligned}
            \ket{\Psi^\prime}&=\cos\frac{\theta}{2}\ket{00}+\sin\frac{\theta}{2}\left(e^{i\varphi_0}\ket{\phi^+_{12}}\ket{0}+e^{i\varphi_1}\ket{0}\ket{\psi^+_{12}}\right)\\
            &+e^{i(\varphi_0+\varphi_1)}\left(e^{i\eta}\ket{\phi^-_{12}\psi^-_{12}}-\cos\frac{\theta}{2}\ket{\phi^+_{12}\psi^+_{12}}\right),
        \end{aligned}
    \end{equation}
    where the orthogonal bases $\left\{\ket{\phi^-_{12}},\ket{\phi^+_{12}}\right\}$ and $\left\{\ket{\psi^-_{12}},\ket{\psi^+_{12}}\right\}$ are written as
    \begin{equation}
        \begin{aligned}
            \ket{\phi^+_{12}}=\cos\frac{\nu_0}{2}\ket{1}+\sin\frac{\nu_0}{2}e^{i\omega_0}\ket{2},\quad&\ket{\phi^-_{12}}=\sin\frac{\nu_0}{2}\ket{1}-\cos\frac{\nu_0}{2}e^{i\omega_0}\ket{2},\\
            \ket{\psi^+_{12}}=\cos\frac{\nu_1}{2}\ket{1}+\sin\frac{\nu_1}{2}e^{i\omega_1}\ket{2},\quad&\ket{\psi^-_{12}}=\sin\frac{\nu_1}{2}\ket{1}-\cos\frac{\nu_1}{2}e^{i\omega_1}\ket{2}.\\
        \end{aligned}
    \end{equation}
    
    To simplify Eq. (\ref{eq:appcomplex}), we construct $U_{12}$ as
    \begin{equation}
        U_{12}=\ket{0}\bra{0}+\ket{1}\bra{\phi^+_{12}}+\ket{2}\bra{\phi^-_{12}}.
    \end{equation}
    Recalling the property that the maximally entangled states remain fully entangled under local unitaries, we apply $U_{12}\otimes U_{12}^{*}$ on the set of target states $\mT$:
    \begin{equation}
        U_{12}\otimes U_{12}^{*}\ket{\Psi_0(\alpha)}\equiv\ket{\Psi_0(\alpha)}, \quad U_{12}\otimes U_{12}^{*}\ket{\Psi^\prime} = \ket{\Psi^{\prime\prime}}.
    \end{equation}
    Notice that the unitary $U_{12}\otimes U^*_{12}$ does not have an effect on $\mS_0$. Thus, we can rewrite a simplified form $\ket{\Psi^{\prime\prime}}$ of $\ket{\Psi^\prime}$ while $\mS_0$ remains untouched
    \begin{equation}
        \begin{aligned}
        \ket{\Psi^{\prime\prime}}&=\cos\frac{\theta}{2}\ket{00}+\sin\frac{\theta}{2}\left(e^{i\varphi_0}\ket{10}+e^{i\varphi_1}\ket{0}\ket{\psi^{\prime+}_{12}}\right)\\
        &+e^{i(\varphi_0+\varphi_1)}\left(e^{i\eta}\ket{2}\ket{\psi^{\prime-}_{12}}-\cos\frac{\theta}{2}\ket{1}\ket{\psi^{\prime+}_{12}}\right),\\
        \ket{\psi^{\prime+}_{12}}=&\cos\frac{\nu}{2}\ket{1}+\sin\frac{\nu}{2}e^{i\omega}\ket{2},\quad\ket{\psi^{\prime-}_{12}}=\sin\frac{\nu}{2}\ket{1}-\cos\frac{\nu}{2}e^{i\omega}\ket{2}.
        \end{aligned}
    \end{equation}
    The parameter $\theta$ has domain $(0,\pi]$ while $\nu\in[0,\pi]$. The phase parameters $\varphi_0,\varphi_1,\omega$, and $\eta$ have domain $[0,2\pi)$.

    After parametrization, we characterize $\mV_\mT$ by constructing the orthogonal projection of $\ket{\Psi^{\prime\prime}}$ to $\spanned\left\{\ket{00},\ket{11}+\ket{22}\right\}$
    \begin{equation}
        \ket{\Psi^{\prime\prime}_\perp}=\ket{\Psi^{\prime\prime}}-\bra{00}\ket{\Psi^{\prime\prime}}\ket{00}-\frac{\bra{11}+\bra{22}}{\sqrt{2}}\ket{\Psi^{\prime\prime}}\frac{\ket{11}+\ket{22}}{\sqrt{2}},
    \end{equation}
    Then we show that $\ket{\Psi^{\prime\prime}}$ and $\ket{\Psi_\perp^{\prime\prime}}$ have one-to-one correspondence in order to prove the existence of the type III structure.

    We now categorize different $\ket{\Psi^{\prime\prime}}$ by splitting the domain of $\nu$ into three parts.
    
    \noindent \textbf{(1)} For $\nu=0$, $\ket{\Psi^{\prime\prime}}$ is parametrized by the tuple $(\theta,\varphi_0,\varphi_1,\eta)$
    \begin{equation}
            \ket{\Psi^{\prime\prime}}=\cos\frac{\theta}{2}\ket{00}+\sin\frac{\theta}{2}\left(e^{i\varphi_0}\ket{10}+e^{i\varphi_1}\ket{01}\right)-e^{i(\varphi_0+\varphi_1)}\left(\cos\frac{\theta}{2}\ket{11}+e^{i\eta}\ket{22}\right),
    \end{equation}
    which leads us to the orthogonal projection
    \begin{equation}
        \ket{\Psi_\perp^{\prime\prime}}=\sin\frac{\theta}{2}\left(e^{-i\varphi_0}\ket{01}+e^{-i\varphi_1}\ket{10}\right)-\frac{1}{2}\left(\cos\frac{\theta}{2}-e^{i\eta}\right)\left(\ket{11}-\ket{22}\right).
    \end{equation}
    Notice that this formula follows the same pattern of Eq. (\ref{eq:pw}) in Appendix \ref{app:A} while we set $\lambda_1=\lambda_2=1$.  
    Similarly, $\mT$ takes the type II form for $\nu=0$.

    \noindent \textbf{(2)} For $\nu\in(0,\pi)$, $\ket{\Psi^{\prime\prime}}$ is parametrized by the tuple $(\theta,\nu,\varphi_0,\varphi_1,\omega,\eta)$
    \begin{equation}
        \begin{aligned}
            &\ket{\Psi^{\prime\prime}}=\cos\frac{\theta}{2}\ket{00}+\sin\frac{\theta}{2}\left(e^{i\varphi_0}\ket{10}+e^{i\varphi_1}\cos\frac{\nu}{2}\ket{01}+e^{i(\varphi_1+\omega)}\sin\frac{\nu}{2}\ket{02}\right)\\
            &+e^{i(\varphi_0+\varphi_1)}\left[e^{i\eta}\left(\sin\frac{\nu}{2}\ket{21}-e^{i\omega}\cos\frac{\nu}{2}\ket{22}\right)-\cos\frac{\theta}{2}\left(\cos\frac{\nu}{2}\ket{11}+e^{i\omega}\sin\frac{\nu}{2}\ket{12}\right)\right],\\
        \end{aligned}
    \end{equation}
    which leads us to a unique orthogonal projection $\ket{\Psi^{\prime\prime}_\perp}$
    \begin{equation}
        \begin{aligned}
            \ket{\Psi^{\prime\prime}_{\perp}}&=\sin\frac{\theta}{2}\left(\ket{10}+e^{i(\varphi_1-\varphi_0)}\cos\frac{\nu}{2}\ket{01}+e^{i(\varphi_1-\varphi_0+\omega)}\sin\frac{\nu}{2}\ket{02}\right)\\
            &+e^{i\varphi_1}\sin\frac{\nu}{2}\left(e^{i\eta}\ket{21}-e^{i\omega}\cos\frac{\theta}{2}\ket{12}\right)-\frac{1}{2}e^{i\varphi_1}\cos\frac{\nu}{2}\left(\cos\frac{\theta}{2}-e^{i(\eta+\omega)}\right)\left(\ket{11}-\ket{22}\right).
        \end{aligned}
    \end{equation}
    This one-to-one correspondence between $\ket{\Psi^{\prime\prime}}$ and $\ket{\Psi^{\prime\prime}_\perp}$ ensures that there does not exist a second single state in $\mT$. Thus, we arrive at the type III form for $\nu\in(0,\pi)$.
    
    \noindent \textbf{(3)} For $\nu=\pi$, $\ket{\Psi^{\prime\prime}}$ is parametrized by the tuple $(\theta,\varphi_0,\varphi_1,\eta)$
    \begin{equation}
        \ket{\Psi^{\prime\prime}}=\cos\frac{\theta}{2}\ket{00}+\sin\frac{\theta}{2}\left(e^{i\varphi_0}\ket{10}+e^{i\varphi_1}\ket{01}\right)-e^{i(\varphi_0+\varphi_1)}\left(\cos\frac{\theta}{2}\ket{12}+e^{i\eta}\ket{21}\right),
    \end{equation}
    which also leads us to a unique orthogonal projection $\ket{\Psi^{\prime\prime}_\perp}$
    \begin{equation}
        \ket{\Psi^{\prime\prime}_\perp}=\sin\frac{\theta}{2}\left(e^{-i\varphi_1}\ket{10}+e^{-i\varphi_0}\ket{01}\right)-\left(\cos\frac{\theta}{2}\ket{12}+e^{i\eta}\ket{21}\right).
    \end{equation}
    It follows that $\mT$ takes the type III form for $\nu=\pi$.

    In summary, we have obtained all of the possible states $\ket{\Psi^{\prime}}\in\mT\setminus\mS_0$. Therefore, we can conclude that there are only three types of $\mT$ in the completely degenerate case.
    \twocolumngrid
    \bibliography{masking}

\end{document}